\documentclass[english, a4paper,numberwithinsect]{article}
\newif\ifejcolor
\ejcolortrue

\usepackage{color}
\usepackage{amsmath}
\usepackage{amssymb}
\usepackage{amsthm}
\usepackage{stmaryrd}
\usepackage{xypic}
\usepackage{pgf}
\usepackage{tikz}
\usetikzlibrary{patterns,snakes}

\newcommand\Z{\mathbb{Z}}
\newcommand\N{\mathcal{N}}

\title{Infinite Communication Complexity}
\author{Pierre Guillon\\
I2M, UMR 7373 - Campus de Luminy\\
Avenue de Luminy - Case 907\\
13288 MARSEILLE Cedex 9, FRANCE\\
\texttt{pierre.guillon@univ-amu.fr}
\and
Emmanuel Jeandel\\
LORIA, UMR 7503 - Campus Scientifique, BP 239\\
54\,506 VANDOEUVRE-L\`ES-NANCY, FRANCE\\
\texttt{emmanuel.jeandel@loria.fr}}

\theoremstyle{plain}
\newtheorem{theorem}{Theorem}
\newtheorem{definition}{Definition}[section]
\newtheorem{proposition}[definition]{Proposition}

\newtheorem{cor}[definition]{Corollary}

\newtheorem{conj}{Conjecture}

\begin{document}

\maketitle
\begin{abstract}

Suppose that Alice and Bob are given each an infinite string, and they
want to decide whether their two strings are in a given relation. How
much communication do they need? How can communication be even defined
and measured for infinite strings? In this article, we propose a formalism
for a notion of infinite communication complexity, prove that it satisfies
some natural properties and coincides, for relevant applications, with
the classical notion of amortized communication complexity. Moreover,
an application is given for tackling some conjecture about
tilings and multidimensional sofic shifts.

\end{abstract}

\section{Introduction}
In this article, we are interested in introducing a generalization of
Communication Complexity \cite{Yao,Comm} to infinite inputs.
The motivation comes from the theory of tilings and regular languages
of infinite pictures, specifically to express that two neighbour cells
can exchange only a finite information.

In this setting, Alice and Bob have a binfinite
word $x$ and $y$. We think of this input as an infinite array of
cells, each containing a symbol, the cell $x_i$ having a channel
of communication with the cell $y_i$. We are looking at decentralized
protocols, that is each decision must be made locally.
These precautions are mandatory to avoid unrealistic protocols. As an
example, a protocol that sends 2 bits in each channel can be simulated
by a protocol that sends only one bit in each channel: instead of
sending, in each channel $i$, two bits $u_i$ and $v_i$, just send the
bit $u_i$ in channel $2i$ and the bit $v_i$ in channel $2i+1$.
Our definition will implicitely forbid such a protocol, and provide a
meaningful definition of ICC.
This situation is similar to what happens for an other model of
communication complexity with infinite inputs called  algebraic
communication complexity \cite{DisComm,DisComm2,DisComm3}, where coding two real numbers
into a single real number should also be forbidden.

We will focus in this article only on nondeterministic communication
complexity. It is the most natural from the point of view of
applications, and is also the easiest to define.
We postpone definitions of deterministic (and probabilistic)
communication complexity to further articles.

Once the definition is given, we will see that many propositions from
finite communication complexity can be translated, with different
proofs, to the infinite setting. We will prove in particular that this
new notion of complexity coincides for relations with what is known as
amortized complexity \cite{AmorComm}. Amortized complexity asks for the best
protocol to decide whether $n$ pairs $(x_i, y_i)$ belong to some
relation $R$, for large $n$, while infinite complexity asks for the best
protocol when the number of pairs is infinite. It is natural that
these quantities should be equal.

While the definition is interesting in its own right, the main
motivation comes from the theory of tilings, where the concept is
quite natural when dealing with regular languages of
infinite pictures. In the last section we will see how this new tool
gives us new insights into this theory.

\section{Preliminaries}

\subsection{Communication Complexity}

We introduce here the formalism of (nondeterministic) Communication
Complexity, with a few innocuous adjustments that will make the
transition to infinite communication complexity easier.
The first chapters of~\cite{Comm} are recommended readings.

Let $S \subseteq X \times Y$ be a binary relation (in communication
complexity, we usually think of $S$ as a boolean function).
Alice is given $x \in X$ and Bob is given $y \in Y$ and they want to
know whether $(x,y) \in S$. The communication complexity of $S$ is the
number of bits that Alice and Bob need to exchange to decide whether
$(x,y) \in S$.

In nondeterministic communication complexity, Alice and Bob are
allowed nondeterministic choices, and must succeed only if
 $(x,y) \in S$.
As in the definition of NP, an alternative definition can be given in
terms of proofs: Alice and Bob are both given nondeterministically a
``proof'' $z$ that $(x,y) \in S$ and each of them uses $z$ to verify
that indeed $(x,y) \in S$.
This is defined formally as follows:

\begin{definition}
A \emph{nondeterministic protocol} for a relation $S \subseteq X \times Y$
is a tuple $(Z,S_X, S_Y)$ where $S_X \subseteq X \times Z$, $S_Y
\subseteq Y \times Z$ and 
\[(x,y) \in S \iff \exists z \in Z, (x,z) \in S_X, (y,z) \in S_Y~.\]
The \emph{size} of the protocol is $\log |Z|$.
The \emph{nondeterministic communication complexity} of $S$, denoted
$\N(S)$ is the minimal size of a protocol for $S$.
\end{definition}
We now give a few examples that will be useful later on.

\begin{itemize}
\item Let $X = Y = \{0,1\}^n$ and $EQ = \{ (x,y) | x = y \}$.
  Then $\N(EQ) = O(n)$. Intuitively, Alice sends all of her
  bits to Bob. (Formally, take $Z = \{0,1\}^n$ and $S_X = S_Y = EQ$).
  It can be proven that this bound is tight: $\N(EQ) = n$.
\item Let $X = Y = \{0,1\}^n$ and $NEQ = \{ (x,y) | x \not= y \}$.
  Then we have $\N(NEQ) = O(\log n)$. Intuitively, Alice chooses $i$
  and sends both $i$ and $x_i$ to Bob, who verifies that $y_i \not= x_i$. 
  This takes $\log n + 1$ bits.  
  (Formally, take $Z = \llbracket 1,n\rrbracket \times \{0,1\}$,
  $S_X = \{ (x,(i,a)) | x_i = a \}$ and
  $S_Y = \{ (y,(i,a)) | y_i \not= a \}$).
  It can be proven that this bound is almost tight: $\mathcal{N}(NEQ) \geq \log n$.
  
\end{itemize}

\subsection{Symbolic Dynamics}

We will now generalize this definition to take into account infinite
inputs, that is inputs in $A^\mathbb{Z}$ for some finite set $A$.
Choosing biinfinite words rather than infinite words is of no
consequence but simplifies the exposition.

The idea is that Alice and Bob both have an infinite word as input,
which can be thought of as an infinite collection of cells.
There is a channel of communication between the cell $i$ of Alice and
the cell $i$ of Bob, and Alice and Bob should use these channels to
communicate.

Now we have to be careful with the exact definition of a protocol.
Consider the case of the problem $EQ$ where $X = Y = A^\mathbb{Z}$ for some finite alphabet $A$. 
The ``optimal'' protocol for $EQ$ should be for Alice something like
sending $x_i$ through the $i$-th channel $i$, and for Bob to compare
the output of its $i$-th channel with $y_i$. This would use $\log |A|$
bits per channel.

However, other protocols are possible: $A^\mathbb{Z}$ is in bijection
with $\{0,1\}^\mathbb{Z}$ so another protocol would be for Alice to
transform its word $x \in A^\mathbb{Z}$ into a word $f(x) \in
\{0,1\}^\mathbb{Z}$, sending $f(x)_i$ through the $i$-th channel, then
for Bob to apply $f^{-1}$ on the whole word it receives through all
channels, then compare the output with $y$. This would use only $1$
bit per channel.

This protocol is of course not what a good protocol should be and it
will forbidden by the definitions. We will ask for all cells of Alice
to act in the \emph{exact same way}, and for the communication on a
given cell to depend only on finitely many cells of the input.

The best way to formalize all this convincingly is with the vocabulary
and the formalism of \emph{symbolic dynamics}. Indeed, the first
property corresponds to an invariance by translation, and the second
property to a continuity argument, both being central in the study of
symbolic dynamics. We refer the reader to \cite{LindMarcus} for a
good introduction to this domain.

Using the definitions below, we will be able to answer the three
following questions:
\begin{itemize}
	\item What relations $S$ should be considered ?
	\item What should be $Z$, $S_X$, $S_Y$ ? (What is a protocol ?)
	\item How do you measure the size of a infinite set ? (What is the complexity ?)	
\end{itemize}	

The basic infinite sets we will be considering are called
\emph{subshifts}. Formally speaking, a subshift is a topologically closed subset which is invariant by the shift map $\sigma$, defined by $\sigma(x)_i=x_{i+1}$ for $x\in A^\Z$ and $i\in\Z$. This encompasses both desired properties: Every cell
behaves the same (shift-invariance), and operations depend on finitely
many cells (continuity, here in the form of closedness/compactness).
We will use here the following, equivalent, definition: A set $S \subset
A^\mathbb{Z}$ is a subshift (or simply shift) if there exists a set of words $\mathcal{F}$
over $A$ so that $S$ is exactly the set of infinite words that do not
contain any pattern in $\mathcal{F}$ as a factor.

For example, the set of biinfinite words over $A = \{a,b\}$ that contains at most
one symbol $b$ is a subshift, corresponding to $\mathcal{F} = \{ ba^nb, n \in \mathbb{N}\}$.

Intuitively, Alice can semi-decide if an biinfinite word $w$ is in a
subshift $S$: on each cell $i$, the same program is executed that reads
continuously the letters around $i$, and fails if it sees a forbidden pattern.
If there is one, some cell will fail at some time $t$ and every
cell will fail at some time.

This might seem to be too powerful (Alice doesn't even need to be
computable in this definition). There are various classes of subshifts
that might be expressed in terms of restriction of Alice's power.
We will see the class of sofic shifts later on, but we will focus here
on \emph{shifts of finite type}.
A subshift $S$ is of finite type if it can be given by a finite
set $\mathcal{F}$ of words of some size $n$.
This corresponds to the case where in each cell $i$ the same program is
run for a given time $t$ (reading the content of cell
$i$ and adjacents cells) and then the word is accepted if none
of the programs has failed by time $t$.


With this formalism, we can now describe what a relation and a
protocol are.
It remains to define  how to measure the size of the sets.
The good notion for this is \emph{entropy} \cite{LindMarcus}.
Informally, a subshift $S$ has entropy $\log c$ if it has $\Omega(c^n)$
differents factors of length $n$. 

Formally, if we denote by $c_n$ the number of different words of size
$n$ of $S$, then the entropy $H(S)$ can be defined by:
\[ H(S) = \lim_n \frac{\log c_n}{n}~.\]
As an exemple $\{0,1, 2 \dots c\}^\mathbb{Z}$ has entropy $\log c$.
The entropy is a good notion of complexity, as is made clear by the
following remarks. First, if $S \subseteq S'$ then $H(S) \leq H(S')$.
Second, if $S$ maps onto $S'$, then $H(S) \geq H(S')$. More precisely:
\begin{definition}
A \emph{block code} $f$ is a continuous, shift-commuting map
$f: S \rightarrow S'$.
For such a map, $H(f(S)) \leq H(S)$.
If $f$ is one-to-one, then $H(f(S)) = H(S)$.

We say that $S$ factors onto $S'$ (and that $S'$ is a factor of $S$)
if there is an \emph{onto} block code $f: S \rightarrow S'$ (also called a factor map)
If $f$ is also one-to-one, $f$ is called a conjugacy, and we say that
$S$ and $S'$ are conjugated.

In particular, if $S'$ is a factor of $S$ then $H(S') \leq H(S)$.
If $S'$ is conjugated to $S$ then $H(S') = H(S)$.

\end{definition}
These few properties mean that reasoning on the size of finite sets
may be translated easily into statements on entropy of shifts.
A notable difference is that $H(S) < H(S')$ does not imply
that there is a one-to-one map from $S$ to $S'$, or an onto map from
$S'$ to $S$.

\subsection{Definition}

We are now ready to define nondeterministic communication complexity:
\begin{definition}
A \emph{nondeterministic protocol} for a subshift $S \subseteq X \times Y$
is a tuple $(Z,S_X, S_Y)$ where $S_X \subseteq X \times Z$, $S_Y
\subseteq Y \times Z$ are subshifts and for all $(x,y) \in X \times Y$,
\[(x,y) \in S \iff \exists z \in Z, (x,z) \in S_X, (y,z) \in S_Y~.\]
The \emph{size} of the protocol is $H(Z)$.
The \emph{nondeterministic communication complexity} of $S$, denoted
$\N(S)$ is the infimum of the size of a protocol for $S$.
\end{definition}

As explained above, the definition mirrors the one in the
finite case, replacing ``finite set'' by ``subshift'' and ``size'' by
``entropy''.

We will see in the next section that this is the good definition to
adopt, as natural and obvious statements will indeed be true.

In the remaining of the paper, we will assume that if $S \subseteq X
\times Y$ is a subshift for which we want to compute the communication
complexity, then $X = \{ x | \exists y \in Y, (x,y) \in S\}$ (that is
the map $S \rightarrow X$ is onto), and similarly for $Y$.
This is an innocuous hypothesis but necessary for theorems below not
to fail for stupid reasons. Many statements also assume implicitely that $S$ is
nonempty.

With these hypotheses, a protocol $(Z,S_X,S_Y)$ entails a few maps.
Denote by $L$  the set of triples $(x,y,z)$ such that $z$ is a protocol 
for $(x,y)$, that is $(z,x,y)$ satisfies simultaneously $z \in Z, (x,z) \in S_X, (y,z) \in S_Y$.

Many properties below may be deduced from the following diagram:
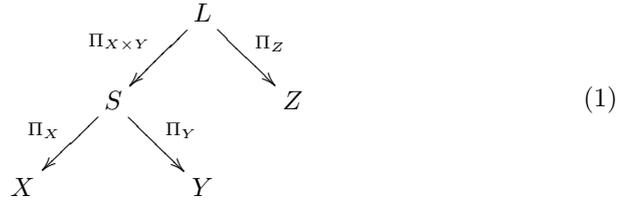
\begin{equation}
\label{eqn:diagram}
\vcenter{	
 \xymatrix@R=20pt@C=20pt{
 && L \ar[dl]_{\Pi_{X\times Y}} \ar[dr]^{\Pi_Z} \\
 &S     \ar[dl]_{\Pi_X} \ar[dr]^{\Pi_Y}&&
 Z\\
 X&&Y
 }}
\end{equation}
 
By definition of a protocol,
$\Pi_{X \times Y}$ is always onto, and we may suppose without loss of
generality that the three other maps involved in this diagram are also
onto (hence factor maps).
This diagram may be completed by maps from/to $S_X$ and $S_Y$ but they
will not be needed explicitely in the following sections.

\section{Properties}

In this section, we will give a few properties of the infinite
communication complexity.  

First a few obvious properties:
\begin{proposition}
	Let $S \subseteq X \times Y$.
\begin{itemize}	
\item	 $\N(S) \geq 0$ (unless $S$ is empty);
\item	$\N(S) \leq \min(H(X),H(Y))$;
\item If $X'$ and $Y'$ are subshifts, then
  $\N(S \cap (X' \times Y')) \leq \N(S)$.
\end{itemize}	
\end{proposition}
\begin{proof}
\begin{itemize}	
\item	Let $Z$ be a protocol for $S$. If $Z$ is empty (and then $H(Z) =
	-\infty$), $S$ will be empty. Otherwise, $Z$ is nonempty and then
	$H(Z) \geq 0$. Therefore $\N(S) \geq 0$.
\item	
	$\N(S) \leq H(X)$ is clear: Take the protocol where Alice sends her
	input to Bob. Formally take $Z = X$, $S_X = \{(x,x) | x \in X\}$ and $S_Y = S$.
\item For the last item, it is clear from the definition that a
  protocol for $S$  may be transformed into a protocol for $S \cap (X'
  \times Y')$ by changing only $S_X$ and $S_Y$.
\end{itemize}
\end{proof}	

The first obvious example requires no communication:
\begin{proposition}
	(If $X$ and $Y$ are nonempty,) $\N(X \times Y) = 0$.
\end{proposition}
\begin{proof}
Alice sends something to Bob, independently of her input. Formally,
take $Z$ consisting only of the periodic point $ \dots 000 \dots$ ($Z$
is of entropy $0$), $S_X = X \times Z$ and $S_Y = Y \times Z$.
This proves $\N(X \times Y) \leq 0$.
Entropy is negative only when $Z$ is empty, hence $\N(S) \geq 0$.

\end{proof}

The first interesting example is equality: Give Alice and Bob each a word
$x$ and $y$, and decide if $x = y$.

\begin{definition}
	If $T$ is a subshift, then $EQ_T = \{ (t,t) | t \in T \}$.
\end{definition}	

\begin{proposition}
	$\N(EQ_T) = H(EQ_T) = H(T)$.
\end{proposition}
\begin{proof}
$\N(EQ_T) \leq H(T)$ is clear: Just take the protocol where Alice sends
its input to Bob. 

Conversely, let $(Z,S_X,S_Y)$ be a protocol for $EQ_T$.
It is clear that in Diagram~\ref{eqn:diagram}, the map $\Pi_Z$ should
be one-to-one: A word $z$ cannot be a protocol for two different pairs
$(x,x)$ and $(y,y)$ as this would imply $(x,y) \in EQ_T$.
As $\Pi_Z$ is one-to-one, this implies $H(Z) = H(L) \geq EQ_T$. 
\end{proof}
We will now give three other proofs of the previous proposition,
introducing other methods to give lower bounds for Communication
Complexity. 

First is the well-known method of fooling sets.
\begin{definition}
	Let $S \subseteq X \times Y$ be a subshift.

	A \emph{fooling set} is a subshift $F \subseteq S$ such that:
 For each $(x,y) \in F$,  there exist at most countably
		  many pairs $(x',y') \in F$ so that $(x,y') \in S$ and $(x',
		  y)\in S$.
\end{definition}	
The usual definition in the finite case replaces ``at most countably many pairs'' by ``no other
pair''.
In the infinite case, we can obtain a stronger statement.

\begin{theorem}
	Let $F$ be a fooling set for $S$. 
	Then $\N(S) \geq H(F)$.
\end{theorem}
As $EQ_T$ is a fooling set for $EQ_T$, this gives a proof of the
previous proposition.

\begin{proof}  
Let $(Z,S_X,S_Y)$ be a protocol for $S$.  Let $L_F$ be the restriction
of $L$ to tuples $(z,x,y)$ where $(x,y) \in F$. We now look at the
diagram:
\[
\begin{array}{ccccc}
  & \Pi_{X \times Y} && \Pi_Z & \\
F &\leftarrow &L_F &\rightarrow& Z\\
(x,y)& \mapsfrom &(z,x,y) &\mapsto &z\\
\end{array}
\]

We suppose wlog that the map $\Pi_Z$ in the preceding diagram is onto
(replace $Z$ by $f(Z)$). Now the hypothesis implies that each $z \in Z$ has only
countably many preimages.
This means that the map $\Pi_Z$ is a countable-to-one factor map, which
implies that $H(L) = H(Z)$ \cite{NewtonParry}.
As $\Pi_Z$ is onto, we have $H(Z) \geq H(F)$. The result follows.
\end{proof}
We now relate the communication complexity with the largest set that
can be extracted simultaneously from $X$ and  $Y$.

A \emph{common factor} of $X$ and $Y$ is a factor $F$ from $X$ and
$Y$, by maps $\phi$ and $\psi$ such that the following diagram commutes:
$$
 
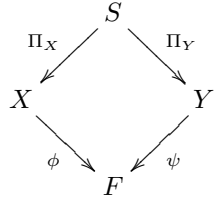
\xymatrix@R=20pt@C=20pt{
 & S \ar[dl]_{\Pi_X} \ar[dr]^{\Pi_Y} \\
 X      \ar[dr]_{\phi}&&
 Y \ar[dl]^{\psi}\\
 &F
 }
 $$

\begin{theorem}
	Let $S \subseteq X \times Y$.
	If $F$ is a common factor of $X$ and $Y$, then $\N(S) \geq H(F)$.
	
	More precisely, if $Z$ is a protocol for $S$, then $Z$ factors onto $F$.
\end{theorem}
As $T$ is a common factor of $EQ_T$, this gives again a new proof of
the proposition.
\begin{proof}
	Let $\theta=\phi\Pi_X = \psi\Pi_Y : S \rightarrow F$ denote the common map and $(Z, S_X,
	S_Y)$ be a protocol for $S$.
	
	Recall the following diagram, where we assume wlog that $\Pi_Z$ is onto.
	$$
		\begin{array}{ccccc}
			& \Pi_{X\times Y} && \Pi_Z & \\
			S &\leftarrow &L &\rightarrow& Z\\
			(x,y)& \mapsfrom &(z,x,y) &\mapsto &z\\
		\end{array}
	$$
	$\theta$ can be lifted to a map $\tilde\theta$ from $L$ to $F$
	by $\tilde\theta = \theta \Pi_{X\times Y}$.
			Now it is easy to see that $\tilde\theta$ depends only on
	$z$.
	Indeed, suppose that $(x,y,z) \in L$ and $(x',y',z)\in L$. By
	definition of a protocol, we also have $(x,y',z) \in L$.
	Now $\tilde\theta(x,y,z) = \phi(x) = \psi(y)$,
	$\tilde\theta(x,y',z) = \phi(x) = \psi(y')$, and
	$\tilde\theta(x',y',z) = \phi(x') = \psi(y')$ 
	from which it follows $\tilde\theta(x,y,z) = \tilde\theta(x',y',z)$.
	This means that $h = \tilde\theta \Pi_Z^{-1}$ is actually a function,
	which is obviously shift-invariant, onto, and standard topological
	arguments show it is continuous.	
	Hence we have a factor map from $Z$ to $W$, and $H(Z) \geq H(W)$.
\end{proof}	

An additional lower bound can be made easily, using the notion of \emph{conditional
  entropy}. We adapt slightly the definitions from \cite{Bowen} to make them work better
in our context.

Let $S \subseteq X \times Y$ be a subshift.
For $x \in X$, let $S^{-1}(x)$ be the set of words $y$ such that $(x,y)
\in S$ and $c_n(x)$ be the number of different words of size $n$ that
can appear in position $[0,n-1]$ in a word in $S^{-1}(x)$.
\clearpage
Then  we can define:
\[H_S(Y|x) = \lim\sup_n \frac{\log c_n(x)}{n}~;\]
\[H_S(Y|X) = \sup_{x \in X} H_S(Y|x)~.\]   

Hence $H_S(Y|X)$ measures somehow how many different words $y$ may
correspond to a given word $x\in X$.
Let us give an example. Let ${LEQ = \{ (x,y) \in \{0,1\}^\mathbb{Z} 	\times \{0,1\}^\mathbb{Z} | \forall i, x_i \leq y_i\}}$.
Then $H_{LEQ} (Y|x) = \log 2$ if $x = \cdots 000 \cdots$, 
$H_{LEQ} (Y|x) = 0$ if $x = \cdots 111 \cdots $, and 
$H_{LEQ} (Y|x) = \log 2 / 2$ if $\forall i, x_i = i \mod 2$.
It is easy to see that $H_{LEQ} (Y|X) = \log 2$.

\begin{theorem}
	Let $S \subseteq X \times Y$.	
	Then $\N(S) \geq H(Y) - H_S(Y|X)$.
\end{theorem}	
For $S = EQ_T$, $H_S(Y|X) = 0$, which gives again a new proof of the
proposition.
\begin{proof}
	First let us detail the proof in the finite case.
	In this case $s(Y|X)$ denotes the maximum number of different $y$
	that can be associated to a given $x$.
	
	Let $(Z,S_X, S_Y)$ be a
	protocol. Now we can enumerate $Y$ like this: first choose some
	$z$ arbitrarily. For this $z$, choose a unique $x$ so that $(x,z)
	\in S_X$. Then enumerate all words $y$ such that $(x,y) \in S$.
    Now the number of $y$ that are enumerated for a given $z$ is less
	than $s(Y|X)$, hence $|Y| \leq |Z| |s(Y|X)|$
	
	Now, for the infinite case.
	Let $(Z,S_X, S_Y)$ be a protocol.
	By properties of the conditional entropy \cite{Bowen},
	$H(S_Y) \leq  H_{S_Y}(Y|Z) + H(Z)$.
    Recall that $H_{S_Y}(Y|Z) = \sup_{z \in Z} H_{S_Y} (Y|z)$ and let
	$z \in Z$.
	There exists $x$ such that $(x,z) \in S_X$.
	Now if $y$ is such that $(y,z) \in S_Y$	then $(x,y) \in S$.    
	This implies that $H_{S_Y}(Y|z) \leq H_S(Y|x)$.
	Hence $H_{S_Y}(Y|z) \leq H_S(Y|X)$, and finally $H_{S_Y}(Y|Z) \leq H_S(Y|X)$.
	
	We obtain $H(S_Y) \leq H_S(Y|X) + H(Z)$, hence $H(Z) \geq H(S_Y) -	H_S(Y|X) \geq H(Y) - H_S(Y|X)$.
\end{proof}

We now go back again specifically to the equality example.
If we look at the propositions above, we can conclude that 
if $(Z,S_X, S_Y)$ is a protocol for $EQ_T$, then 
there exist $Z' \subseteq Z$ and a factor map from $Z'$ to $T$.

We look now at $T = \{3,4\}^\mathbb{Z} \cup \{5,6\}^\mathbb{Z}$.
$T$ is of entropy $\log 2$. However, there is no factor from $Z =
\{0,1\}^\mathbb{Z}$ to $T$.
That is,  there are no protocols $(Z,S_X,S_Y)$ for $EQ_T$ where $Z =
\{0,1\}^\mathbb{Z}$. This means that having a protocol with $H(Z) =
\log 2$ is not the same as having a protocol with $Z =
\{0,1\}^\mathbb{Z}$.
This is however an artifact of protocols of exact communication
complexity and this does not happen otherwise, as illustrated as follows.

We now introduce a class of shifts, parametrized by $\beta$, 
called the $\beta$-shifts \cite{Parry,Sigmund}, 
which have the following properties: The $\beta$-shift has entropy
$\log \beta$, and for $\beta \in \mathbb{N} \setminus \{0,1\}$, the
$\beta$-shift coincides with the full shift $\{0,1\dots\beta\}^\mathbb{Z}$.
The exact definition is not important, but let just note that the
$\beta$-shift corresponds somehow to numeration in the (possibly
nonintegral) base $\beta$.

Then we can prove
\begin{theorem}
	Let $S \subseteq X \times Y$ be a subshift.
	
	For any $\beta > \N(S)$, there is a
	protocol $(Z,S_X, S_Y)$ where
	$Z$ is the $\beta$-shift.
\end{theorem}	
The preceding discussion shows that the result is not true for $\beta = \N(S)$.

\begin{proof}
The idea is to use Krieger's embedding theorem \cite{LindMarcus},
which states that any subshift $S$ can be embedded into a subshift $T$
provided that $T$ has bigger entropy, has more periodic points, and
satisfies a technical condition (be a mixing SFT).
The set of $\beta$ for which this technical condition is true is dense
in $[1,+\infty[$ \cite{Parry}.

Let $\beta > \N(S)$.	
By definition of $\N(S)$, there exists a protocol $(Z,S_X, S_Y)$ for
$S$ where $H(Z) < \beta$.

By changing $Z$ to a product of $Z$ and a Thue-Morse shift, we may
assume wlog that $Z$ has no periodic point.

Now we can find $H(Z) <\beta' <  \beta$ so that the $\beta'$-shift is
a (mixing) SFT. As $Z$ has no periodic points, we can use Krieger's
embedding theorem to embed $Z$ into the $\beta'$-shift, hence into
the $\beta$-shift.

Now we have obtained a protocol $(Z,S_X, S_Y)$ where $Z$ is included
in the $\beta$-shift.
Replacing $S_X$ by $S_X \cap (A \times Z)$, and similarly for $S_Y$,
we may replace $Z$ by the whole $\beta$-shift and then obtain the
theorem.
\end{proof}	
\section{Amortized Communication Complexity}
\label{sec:amor}
In this section, we give a link between infinite communication
complexity and asymptotic communication complexity.
Let $R \subseteq X \times Y$ be a relation.
We denote by $R^n \subseteq X^n \times Y^n$ the relation 
$(x,y) \in R^n \iff \forall i<n, (x_i, y_i) \in R$.

\begin{definition}[\cite{AmorComm}]
	The asymptotic (amortized) communication complexity of $R$ is 
	\[ \N^{asymp}(R) = \lim_n \frac{\N(R^n)}{n}~.\]
\end{definition}

With the same notation we denote by 
$R^\mathbb{Z} \subseteq X^\mathbb{Z} \times Y^\mathbb{Z}$ the relation
$(x,y) \in R^\mathbb{Z} \iff \forall i \in \mathbb{Z}, (x_i, y_i) \in R$.

\begin{theorem}\label{t:asymp}
	$\N(R^\mathbb{Z}) = \N^{asymp}(R)$.
\end{theorem}
In other words, the asymptotic complexity is the same as the infinite
complexity.

Before proving the result, we need a related proposition, that states
that subshifts with simple description admit protocol with simple
descriptions:
\begin{proposition}
Let $S \subseteq X \times Y$ be a subshift of finite type.
If $(Z, S_X, S_Y)$ is a protocol for $S$, then for any $\epsilon$, 
there exists a protocol $(Z', S'_X, S'_Y)$ for $S$ where $Z', S'_X$
and $S'_Y$ are also subshifts of finite type.
\end{proposition}	
\begin{proof}
Let $(Z, S_X, S_Y)$ be a protocol for $S$.

$Z, S_X$ and $S_Y$ are defined by families of forbidden patterns, denoted
by $Z^n$, $S^n_X$ and $S^n_Y$, the subshifts forbidding only the first $n$
patterns. Hence $S_X = \cap_n S^n_X$, and similarly for $Z, S_Y$.

If we do the same protocol with $S^n_X$ instead of $S_X$, we will
recognize a superset of $S$ that we call $S^n$.
Let us prove $S = \cap_n S^n$, one inclusion being
obvious. If $(x,y)\in
\cap_n S^n$, then there exists $z_n \in Z^n$ such that $(x, z_n) \in
S^n_X$ and $(y, z_n) \in S^n_Y$.
By compactness, the sequence $(z_n)$ admits a limit point $z \in Z$,
that satisfies $(x,z) \in S_X$ and $(y,z) \in S_Y$, hence $(x,y) \in S$.

But $S$ is supposed to be of finite type, so defined by finitely many
patterns. At some point $n_0$ all these patterns will be forbidden,
that is $S = \cap_{n < n_0} S^n$.
But this means that for all $n \geq n_0$,  $(Z^n, S^{n}_X, S^{n}_Y$) is a protocol
for $S$, for which all subshifts involved are of finite type.

Now, as entropy is upper-semicontinuous for shift spaces, $H(Z) =
\lim_n H(Z_n)$, hence for $n \geq n_0$ big enough, we will have
$H(Z_n) < H(Z) + \epsilon$ which proves the proposition.
\end{proof}	

\begin{proof}[Proof of Theorem~\ref{t:asymp}]
\begin{itemize}
\item First, let us prove $\N(R^\mathbb{Z}) \leq \N^{asymp}(R)$.
  Let $\epsilon > 0$ and let $(Z, S_{X^n}, S_{Y^n})$ be a (finite!)
	protocol for $R^n$ of complexity at most $n\N^{asymp}(R)+n\epsilon$, that
	is $\log|Z|\leq n\N^{asymp}(R)+n\epsilon$.
	We build a protocol for $R^\mathbb{Z}$ as follows.
	
    Let $Z'$ be the subshift over the language $Z \cup \{\bot\}$
	defined as follows: a word $w$ is in $Z'$ if and only if every
	factor of $w$ of length $n$ contains exactly one letter in $Z$.
	In other words, a word in $Z'$ contains a letter in $Z$, then
	$n-1$ symbols $\bot$, then a letter in $Z$, ad libitum.
	It is clear that $Z'$ is of entropy $\frac{\log|Z|}{n} \leq 	\N^{asymp}(R) + \epsilon$.
		
	We now define $S_{X^\mathbb{Z}}$ as follows: 
    $(x,z) \in 	S_{X^\mathbb{Z}}$ if and only if
	$z \in Z'$ and, if we denote by $I = i+n\mathbb{Z}$ the positions
	in $z$ where the letter is not $\bot$, then for all $j$,
	$(x_{i+jn}x_{i+1+jn}\dots x_{i+n-1+jn}, z_{i+jn}) \in S_{X^n}$.
	We define $S_{Y^\mathbb{Z}}$ in the same way.
	
	It is clear from the definition that we obtain this way a protocol
	for $R^\mathbb{Z}$ of size at most $\N^{asymp}(R) + \epsilon$,
	which gives the result.				
\item $\N(R^\mathbb{Z}) \leq \N^{asymp}(R)$.
  Let $\epsilon > 0$ and let $(Z, S_{X^\mathbb{Z}}, S_{Y^\mathbb{Z}})$
  be a (infinite) protocol for $R^\mathbb{Z}$. 
  As $R^\mathbb{Z}$ is of finite type (it is defined by forbidden
  patterns of size $1$), we may suppose by the
  previous proposition that the protocol is of finite type.
  
  Let $L$ as above be the set of tuples $(z,x,y)$ such that 
  $z \in Z$, $(x,z)\in S_{X^\mathbb{Z}}$, $(y,z) \in S_{Y^\mathbb{Z}}$.
  $L$ is a subshift of finite type (it is the intersection of three
  subshifts of finite type), hence can be defined by a finite set of
  forbidden words, say of size $r$.
  
 Let $n \geq r$. We now describe a protocol for $R^n$.  
 On input $(x,y)$, we send to Alice and Bob a word $z$ of size $n$
 that is valid for $Z$,
 and Alice (resp. Bob) sends to Bob (resp. Alice) its first and last $r$ letters.
 Now Alice looks whether the pattern $(x,z)$ appears in some valid
 word of
 $S_{X^\mathbb{Z}}$, and whether the two patterns of length $r$ that Bob sent to her
 appears in some valid word of $L$.
In this case, Alice accepts. Bob does the same.

Now, if $(x,y) \in R^n$, it is clear that the above protocol works:
just complete $(x,y)$ to obtain an infinite word $(x^\infty,y^\infty$)
with $(x,y)$ in its center in $R^\mathbb{Z}$,
take $z^\infty$ to be the word that proves that $(x^\infty,y^\infty)$
is indeed valid, and take for $z$ the central $n$ letters of $z^\infty$.

Conversely, suppose that $(x,y)$ is accepted by the protocol.
We now look at the word $(x,y,z)$ of size $n$ we obtain.
By construction, this word does not contain any forbidden word of $L$.
Furthermore, its first and last $r$ letters appear in some (possibly
different) infinite words of $L$; this means that we can complete it into an infinite word in $L$.
This proves that there exists an infinite word $(x^\infty, y^\infty,
z^\infty)$ in $L$ with $(x,y,z)$ at its center, hence that $(x^\infty,
y^\infty) \in R^\mathbb{Z}$, which implies that $(x,y) \in R^n$.

If we denote by $c_n$ the number of words of $Z$ of size $n$, then
this protocol uses $\log c_n + 4\log r$ bits, that is
 $\N(R^n) \leq \log c_n + 4 \log r$,  therefore \[\N^{asymp}(R) = \lim_n
 \frac{\N(R^n)}{n} \leq \lim_n \frac{\log c_n}{n} = H(Z) \leq
 \N(R^\mathbb{Z}) + \epsilon~.\]
\end{itemize}
\end{proof}	

\section{Application to 2D languages}

As hinted above, the main motivation comes from the theory of 2D
languages and in particular the definition of regular
languages of infinite pictures, which are called sofic in this context.

The best way to define sofic (2D-)shifts uses the well-known concept of
Wang tiles from tiling theory, as introduced by Hao Wang
\cite{wangpatternrecoII}.
 In our context, a Wang \emph{tile} is a unit square
with colored edges and some symbol $x\in \Sigma$ at its center.
A tiling by a finite set $\tau$ of Wang tiles associates to each point of
the discrete plane $\mathbb{Z}^2$ a Wang tile  so that
contiguous edges have the same color.
By looking at the symbol at the center of each tile, a tiling by $\tau$
gives rises to an infinite picture $w \in \Sigma^{\mathbb{Z}^2}$.
The sofic shift defined by $\tau$ is then the set of infinite pictures we
obtain this way.

A first example is presented in Figure~\ref{fig:wang}. This set of Wang
tiles produces a lot of different tilings but only two different
pictures (up to translation): one with the symbol $0$ everywhere, and
the other one with only one occurrence of the symbol $1$. Hence the set
of pictures over the alphabet $\{0,1\}$ containing at most one
occurrence of the symbol $1$ is a sofic shift.

\newcommand\Wangs[6]{ 
        \draw[line width=0.1pt,fill=#2] (#1) -- +(1,1) -- +(1,3) -- +(0,4) -- +(0,0);
        \draw[line width=0.1pt,fill=#3] (#1) -- +(1,1) -- +(3,1) -- +(4,0) -- +(0,0);
                \begin{scope}[xshift=2cm, yshift=2cm]
                                        \node (#1) {#6};
                \end{scope}
\begin{scope}[xshift=4cm, yshift=4cm]
        \draw[line width=0.1pt,fill=#4] (#1) -- +(-1,-1) -- + (-1,-3) -- +(0,-4) -- +(0,0);
        \draw[line width=0.1pt,fill=#5] (#1) -- +(-1,-1) -- + (-3,-1) -- +(-4,0) -- +(0,0);
\end{scope}     
}

\newsavebox\WangA
\sbox\WangA{\tikz[scale=0.11]{
                \useasboundingbox (-1,1.5) rectangle (5,4);
\ifejcolor              
\Wangs{0,0}{blue}{blue}{blue}{blue}{{\tiny 0}}
\else
\Wangs{0,0}{black!70}{black!70}{black!70}{black!70}{{\tiny 0}}
\fi                             
}}

\newsavebox\WangB
\sbox\WangB{\tikz[scale=0.11]{
                \useasboundingbox (-1,1.5) rectangle (5,4);
\ifejcolor                              
\Wangs{0,0}{blue}{blue}{red}{red}{{\tiny 1}}
\else
\Wangs{0,0}{black!70}{black!70}{black!45}{black!45}{{\tiny\textbf{1}}}
\fi
}}

\newsavebox\WangC
\sbox\WangC{\tikz[scale=0.11]{
				\useasboundingbox (-1,1.5) rectangle (5,4);
\ifejcolor                              
\Wangs{0,0}{blue}{red}{yellow}{red}{{\tiny 0}}
\else
\Wangs{0,0}{black!70}{black!45}{black!20}{black!45}{{\tiny 0}}
\fi                             
}}

\newsavebox\WangD
\sbox\WangD{\tikz[scale=0.11]{
                \useasboundingbox (-1,1.5) rectangle (5,4); 
\ifejcolor                              
\Wangs{0,0}{red}{blue}{red}{yellow}{{\tiny 0}}
\else
\Wangs{0,0}{black!45}{black!70}{black!45}{black!20}{{\tiny 0}}
\fi
}}

\newsavebox\WangE
\sbox\WangE{\tikz[scale=0.11]{
                \useasboundingbox (-1,1.5) rectangle (5,4);
\ifejcolor                              
\Wangs{0,0}{yellow}{yellow}{yellow}{yellow}{{\tiny 0}}
\else
\Wangs{0,0}{black!20}{black!20}{black!20}{black!20}{{\tiny 0}}
\fi
}}

\begin{figure}

\[
        \tau = \left\{
\tikz[scale=0.2]{\useasboundingbox (-1,1.5) rectangle (5,4);  
\ifejcolor
\Wangs{0,0}{blue}{blue}{blue}{blue}{0}
\else
\Wangs{0,0}{black!70}{black!70}{black!70}{black!70}{0}
\fi
}
,
\tikz[scale=0.2]{\useasboundingbox (-1,1.5) rectangle (5,4);
\ifejcolor  
  \Wangs{0,0}{blue}{blue}{red}{red}{1}
\else
  \Wangs{0,0}{black!70}{black!70}{black!45}{black!45}{1}
\fi
  }
,
\tikz[scale=0.2]{\useasboundingbox (-1,1.5) rectangle (5,4);
\ifejcolor
  \Wangs{0,0}{blue}{red}{yellow}{red}{0}
\else
  \Wangs{0,0}{black!70}{black!45}{black!20}{black!45}{0}
\fi
  }
,
\tikz[scale=0.2]{\useasboundingbox (-1,1.5) rectangle (5,4);
\ifejcolor
  \Wangs{0,0}{red}{blue}{red}{yellow}{0}
\else
  \Wangs{0,0}{black!45}{black!70}{black!45}{black!20}{0}
\fi  
  }
,
\tikz[scale=0.2]{\useasboundingbox (-1,1.5) rectangle (5,4);
\ifejcolor  
  \Wangs{0,0}{yellow}{yellow}{yellow}{yellow}{0}
\else
  \Wangs{0,0}{black!20}{black!20}{black!20}{black!20}{0}
\fi
  }
\right\}
\]

\begin{center}
        \begin{tikzpicture}[scale=0.4]
                \useasboundingbox (0,-1) rectangle (10,11);
                \draw (4,-1) node {$B$};
                \clip[decorate, decoration={zigzag,segment length=9mm}] (0.5,0.5) rectangle  (7.5,9.5);
        \foreach \a in {0,1.1,2.2,3.3,4.4,5.5,6.6,7.7,8.8,9.9,11}
                {
                  \foreach \b in {0,1.1,2.2,3.3,4.4,5.5,6.6,7.7,8.8,9.9,11}
                  {
                        \node at (\b cm, \a cm) {\usebox\WangA};                        
                        }
                  }
        \end{tikzpicture}
        \begin{tikzpicture}[scale=0.4]
                \useasboundingbox (0,-1) rectangle (10,11);
                \draw (4,-1) node {$Y$};
                \clip[decorate, decoration={zigzag,segment length=9mm}] (0.5,0.5) rectangle  (7.5,9.5);
        \foreach \a in {0,1.1,2.2,3.3,4.4,5.5,6.6,7.7,8.8,9.9,11}
                {
                  \foreach \b in {0,1.1,2.2,3.3,4.4,5.5,6.6,7.7,8.8,9.9,11}
                  {
                        \node at (\b cm, \a cm) {\usebox\WangE};
                        }
                 }
        \end{tikzpicture}
\end{center}
\begin{center}
        \begin{tikzpicture}[scale=0.4]
                \useasboundingbox (0,-1) rectangle (10,11);
                \draw (4,-1) node {$V_i$};
                \clip[decorate, decoration={zigzag,segment length=9mm}] (0.5,0.5) rectangle  (7.5,9.5);
        \foreach \a in {0,1.1,2.2,3.3,4.4,5.5,6.6,7.7,8.8,9.9,11}
                {
                  \node at (4.4 cm, \a cm) {\usebox\WangC};                 
                  \foreach \b in {0,1.1,2.2,3.3}
                  {
                        \node at (\b cm, \a cm) {\usebox\WangA};
                        }
                  \foreach \b in {5.5,6.6,7.7,8.8,9.9,11}
                  {
                        \node at (\b cm, \a cm) {\usebox\WangE};
                        }
                  }
        \end{tikzpicture}
        \begin{tikzpicture}[scale=0.4]
                \useasboundingbox (0,-1) rectangle (10,11);
                \draw (4,-1) node {$H_j$};
                \clip[decorate, decoration={zigzag,segment length=9mm}] (0.5,0.5) rectangle  (7.5,9.5);
        \foreach \a in {0,1.1,2.2,3.3,4.4,5.5,6.6,7.7,8.8,9.9,11}
                {
                  \node at (\a, 5.5 cm) {\usebox\WangD};
                  \foreach \b in {0,1.1,2.2,3.3,4.4}
                  {
                        \node at (\a cm, \b cm) {\usebox\WangA};
                        }
                  \foreach \b in {6.6,7.7,8.8,9.9,11}
                  {
                        \node at (\a cm, \b cm) {\usebox\WangE};
                        }
                  }
        \end{tikzpicture}
        \begin{tikzpicture}[scale=0.4]
                \useasboundingbox (0,-1) rectangle (10,11);
                \draw (4,-1) node {$C_{i,j}$};
                \clip[decorate, decoration={zigzag,segment length=9mm}] (0.5,0.5) rectangle  (7.5,9.5);
                \foreach \a in {5.5,6.6,7.7,8.8,9.9}
                {
                  \node at (\a, 5.5 cm) {\usebox\WangD};
                  \foreach \b in {0,1.1,2.2,3.3,4.4}
                  {
                        \node at (\a cm, \b cm) {\usebox\WangA};
                        }                 
                  \foreach \b in {6.6,7.7,8.8,9.9}
                  {
                        \node at (\a cm, \b cm) {\usebox\WangE};
                        }
                                          }
                                        
                \foreach \a in {6.6,7.7,8.8,9.9}
                                {
                  \node at (4.4 cm, \a cm) {\usebox\WangC};
                  }
                \node at (4.4cm, 5.5cm) {\usebox\WangB};
                \node at (4.4cm, 4.4cm) {\usebox\WangA};

                \foreach \a in {0,1.1,2.2,3.3}
                {
                  \foreach \b in {0,1.1,2.2,3.3,4.4,5.5,6.6,7.7,8.8,9.9}
                  {
                        \node at (\a cm, \b cm) {\usebox\WangA};
                }
                          \node at (4.4 cm, \a cm) {\usebox\WangA};
                }
        \end{tikzpicture}
\end{center}  
\caption{A set of Wang tiles and the tilings we obtain. There are
  countably many tilings but only 5 of them up to translation. All
  tilings except $C_{i,j}$ correspond to an infinite picture with the
  symbol $0$ everywhere, while the tiling $C_{i,j}$ corresponds to an
  infinite picture with the
  symbol $0$ everywhere except on position $(i,j)$, which contains the
  symbol
  $1$.}
\label{fig:wang}
\end{figure}
The main question we want to tackle is a way to decide whether a
given set of infinite pictures is indeed a sofic (2D-)shift.
For one-dimensional shifts, where sofic shifts can be defined in a
similar way, a shift $S$ is sofic if and only if the set of finite
words it contains is regular.


\begin{proposition}
	\label{prop:shift}
	Let $S$ be a two-dimensional shift over an alphabet $A$.
	
	For $n$ an integer, denote by $L_n(S)$ the set of infinite
	words over $A^n$ that may appear as $n$ consecutives rows in some
	element of $S$.	
    Let	
	\[ R_{n,m}(S) = \{ (x,y) \in L_n(S)  \times L_m(S) | xy \in
		  L_{n+m}(S)\}~.\]
		
If $S$ is sofic, then $\N(R_{n,m}(S) = O(1)$ (independently of $n$ and
$m$).
\end{proposition}
\begin{proof}
Let $\tau$ be the set of Wang tiles that defines $S$.
The protocol is obvious. Alice, on input $x$, chooses a way to tile
its part of the space (that may be extended into a tiling of an entire
half-plane), and sends to Bob what are the colors in their common border.
Then Bob accepts iff he can tile its part of the space (and the rest
of the plane) respecting this border condition.
\end{proof}

The proposition is not a characterisation. It turns out however that
it becomes a characterisation for one-dimensional languages.

\begin{theorem}
	Let $S$ be a one-dimensional shift. Let $L_n(S)$ be the set of words
	of size $n$ of $S$ and	
	\[ R_{n,m}(S) = \{ (x,y) \in L_n(S)  \times L_m(S) | xy \in
		  L_{n+m}(S)\}~.\]		
Then $S$ is sofic iff $\N(R_{n,m}(S)) = O(1)$.
\end{theorem}
\begin{proof}
One direction is clear.
Suppose that $S$ is not sofic. To simplify the exposition, we assume that $S$
is over the alphabet $\{0,1\}$.
Let $L(S)$ be the set of finite words that might appear in $S$.
If $S$ is not sofic, then $L(S)$ is not regular;
This implies that there exists a sequence of words $(u_i)_{i \in
  \mathbb{N}}$ in the language of $S$ such that all residual sets  
  \[ u_i^{-1} S = \{ x \in \{0,1\}^\star| u_i x \in L(S) \}  \] are distinct.	  
  Wlog we may suppose that words in the sequence $(u_i)$ are of increasing size, and that $u_i$ is of size at least $i$.
  
  Let $n$ be an integer, and denote by $k$ the size of $u_n$.
  We will prove that there exists $m$ so that $\N(R_{k,m}) \geq \log\log\log n$, hence $\N(R_\cdot,\cdot)$ is not bounded.
	
  By definition, for any word $u$,  $u^{-1} S = (0u)^{-1} S \cup (1u)^{-1} S$.  
  Hence all residual sets for words of size less than $k$ can be expressed in
  terms of residual sets for words of size $k \geq n$.  
  As $u_1 \dots u_n$ (of size less than or equal to $k$) give $n$ different residual sets, this implies in
  particular that there should be at least $\log n$ different residual
  sets for words of size $k$.
  
  As there are finitely many residual sets for words of size $k$,
  there exists a constant $m$ such that if $u^{-1} S \not= v^{-1}S$,
  for $u,v$ of size $k$, then there exists $w$ of size at most $m$ so
  that $uw \in S \iff vw \not\in S$.  
  As any word in $S$ of size less than $m$ can be prolonged into a
  word of $S$ of size $m$, we may suppose that $w$ is of size exactly $m$.
  
  We therefore have obtained the following: We have a family
   $v_1 \dots v_{\log	n}$ in
  $L_k(S)$ for which if $i \not= j$, there exists $w \in L_m(S)$ so
  that $v_i w \in S \iff v_j w \not\in S$.
  
  This implies in  particular that in a protocol $(Z, S_A, S_B)$ for
  $R_{k,m}$, each $v_i$ must issue a different set of responses $z \in
  Z$, which implies that $2^{|Z|}\geq \log n$.  
  This implies that $\log |Z| \geq \log \log \log n$.
  This is true for any protocol, hence $\N(R_{k,m}) \geq \log \log \log n$.        
\end{proof}

Proposition~\ref{prop:shift} gives some insight into the extension
conjecture.
\begin{definition}
If $S$ is a set of infinite words, let $S^\mathbb{Z}$ denote the set
of infinite pictures, where each row is an element of $S$, different rows possibly corresponding to different elements of $S$.
\end{definition}
\begin{conj}[Extension conjecture]
$S^\mathbb{Z}$ is sofic only if $S$ is sofic.	
\end{conj}	
This conjecture was proven in some particular cases
\cite{Pavlov,MaddenKass} that may be seen as instances of a general
communication complexity argument that we formulate now.

Indeed, suppose that $S^\mathbb{Z}$ is sofic. Then by the previous
proposition, $\N(R_{n,n}(S^\mathbb{Z})) = O(1)$.
But $R_{n,n}(S^\mathbb{Z}) = R_{n,n}(S)^\mathbb{Z}$ which means that results
from the previous section can be applied:
$\N(R_{n,n}(S\mathbb{Z})) = \N^{asymp}(R_{n,n}(S))$.

Now, well-known results from Communication Complexity about direct sums
 permit to estimate $\N^{asymp}(R)$ from
$\N(R)$:
\begin{theorem}[{\cite{FracCov},\cite[Corollary 4.9]{Comm}}]
	\label{thm:directsum}
For any relation $R \subseteq \{0,1\}^m \times \{0,1\}^m$, we
have $\N^{asymp}(R) \geq \N(R) - \log m + O(1)$.
\end{theorem}	

\begin{cor}
	If $S^\mathbb{Z}$ is sofic, then $\N(R_{n,n}(S)) \leq \log |\Sigma|
	\log n + O(1)$.
\end{cor}
We may replace in the theorem the right-hand term by $\log |L_n(S)| +
O(1)$.
(Alice and Bob may reject inputs $(x,y)$ that are not in $L_n(S)
\times L_n(S)$ as they cannot correspond to a valid word $xy \in
L_{2n}(S)$).
This corollary may be seen as a
reformulation of  \cite[Proposition 4.3]{Pavlov} in a different
vocabulary which makes the theorem more natural.

Note also that if $S$ is sofic then $\N(R_{n,n}) = O(1)$,
which means that possible counterexamples to the conjecture entail sets
$S$ for which the communication complexity is low but nonconstant. 

A specific potential counterexample is mentioned in \cite{Pavlov}. Let
$S$ be the subshift whose set of forbidden patterns is ${\cal F} = \{
  ca^kdb^kc\}$. That is, every time the pattern $ca^nb^mc$ appears in
some word of $S$, we must have $n \neq m$.

For this particular example, we have $\N(R_{n,n}(S)) = \log |L_n(S)| + O(1)$.
To simplify the exposition, suppose that Alice's word (of size $n$) ends 
with $wca^kd$, Bob's word begins with $b^mc$ and they want
to know whether $k \neq m$.
We will give a (well-known) protocol of complexity $1 + \log \log n \simeq \log |L_n(S)|$.
The following nondeterministic protocol can be used: 
Alice chooses nondeterministically an integer $i$ between
$1$ and $\log n$ and sends $i$ and the $i$-th bit of $k$ to Bob. Then Bob
tests whether the $i$-th bit of $m$ is different from what Alice sent.

Theorem \ref{thm:directsum} gives only a lower bound and does not
preclude that ${\N^{asymp}(R_{n,n}(S)) \neq O(1)}$. 
However it is well-known in this case {FracCov} that
we have ${\N^{asymp}(R_{n,n}(S)) = O(1)}$, so that 
Proposition~\ref{prop:shift} is not sufficient to treat this particular
counterexample.

\section{Open Problems}
The main open question is related to the definition of Communication
Complexity as an infimum. We do not know whether a protocol of
optimal communication complexity ($\N(S)$) is always possible.

Section~\ref{sec:amor} suggests a link between the infinite
communication complexity of a subshift $S$ and the asymptotic limit of
communication complexity of the analog problem for finite words.
We have specific counterexamples showing the two quantities are not
always equal, but we conjecture that the finite version is an upper bound
for the infinite version.

In terms of tilings, proving the extension conjecture, in particular using
infinite communication complexity, remains open.

\appendix

\end{document}
\section{Multi-valued maps}

\begin{theorem}
Let $S \subseteq X \times Y$ compact so that for all $x$, there exists  $y$ so
that $(x,y) \in S$.

Then for all $x$, for all $\epsilon$ there exists $\delta$ so that if
$d(x,x') < \delta$ then there exists $y,y'$ such that
\begin{itemize}
	\item $(x,y) \in S$;
	\item $(x',y') \in S$;
	\item $d(y,y') < \epsilon$.
\end{itemize}	
\end{theorem}
Proof:
Let $Y$ be the set of images of $x$

Consider $T = \{ x' | \exists (z, z'), (x',z) \in S \wedge z' \in Y
 \wedge d(z,z') < \epsilon\}$.

We will prove that $T$ contains an open neighborhood of $x$.

Otherwise let $x_i$ be a sequence converging to $x$ such that $x_i \not\in T$.

Let $z_i$ so that $(x_i,z_i) \in S$.

By definition, for all $z' \in Y$, $d(z_i,z') \geq \epsilon$.

Let $z$ be a limit point $z_i$.
By closure of $S$, $(x,z) \in S$.
So $z \in Y$ by definition of $Y$.
But this contradicts that $\forall z' \in Y, d(z, z') \geq \epsilon$.

\end{document}